\documentclass[a4paper,10pt]{article}\def\zibreport{1}\def\usecolors{1}

\usepackage{authblk}
\makeatletter
\renewcommand\AB@authnote[1]{\rlap{\textsuperscript{\normalfont#1}}}

\makeatother

\usepackage{ifthen}

\usepackage[T1]{fontenc}
\usepackage[latin1]{inputenc}
\usepackage[english]{babel}

\usepackage{booktabs}
\usepackage{pdflscape}
\usepackage{longtable}
\usepackage[textsize=tiny]{todonotes}

\usepackage{amsthm}

\usepackage{xcolor}
\definecolor{darkgreen}{HTML}{009900}

\usepackage{tikz}
\usepackage{pgfplots}

\usepackage[font=small,labelfont=bf]{caption}
\usepackage[font=footnotesize,labelfont=bf]{subcaption}

\usepackage{float}
\restylefloat{table}

\usepackage{amsmath}
\usepackage{amssymb}

\usepackage{xspace}
\usepackage{xcolor}

\newtheorem{Def}{Definition}

\newtheorem{Rem}[Def]{Remark}
\newtheorem{Thm}[Def]{Theorem}

\newcommand{\ie}{i.e.,\xspace}
\newcommand{\eg}{e.g.,\xspace}
\newcommand{\cf}{cf.\xspace}

\newcommand{\bandb}{branch-and-bound\xspace}

\newcommand{\NP}{\mathcal{NP}}

\newcommand{\Norm}[1]{\Vert #1\Vert}

\newcommand{\floor}[1]{\lfloor #1\rfloor}
\newcommand{\ceil}[1]{\lceil #1\rceil}

\newcommand{\fa}{\text{for all }}

\DeclareMathOperator{\diag}{diag}

\newcommand{\N}{\mathbb{N}\xspace}

\newcommand{\Z}{\mathbb{Z}\xspace}
\newcommand{\R}{\mathbb{R}\xspace}

\newcommand{\TransMat}{P\xspace}
\newcommand{\Bins}{\mathcal{B}\xspace}
\newcommand{\NetFlow}{f}
\newcommand{\Coherence}{g}
\newcommand{\Cluster}{C\xspace}
\newcommand{\Clusters}{\mathcal{C}\xspace}

\newcommand{\coh}{\Coherence\xspace}

\usepackage{algorithm}
\usepackage{algorithmic}
\renewcommand{\algorithmicrequire}{\textbf{Input:}}
\renewcommand{\algorithmicensure}{\textbf{Output:}}

\newcommand{\ourtitle}{Mixed-Integer Programming for Cycle Detection in Non-reversible Markov Processes}

\theoremstyle{definition}
\newtheorem{example}{Example}[section]

\ifthenelse{\zibreport = 1}{
 \let\pdfoutorg\pdfoutput
  \let\pdfoutput\undefined
  \usepackage{zibtitlepage}
  \let\pdfoutput\pdfoutorg

  \ZTPTitle{\ourtitle}
  \ZTPAuthor{Isabel~Beckenbach \hspace{.2cm} Leon~Eifler \hspace{.2cm} Konstantin~Fackeldey \hspace{.2cm} Ambros~Gleixner \hspace{.2cm} Andreas~Grever \hspace{.2cm} Marcus~Weber \hspace{.2cm} Jakob~Witzig}
  \ZTPPreprint
  \ZTPNumber{16-39}
  \ZTPMonth{August}
  \ZTPYear{2016}
  \ZTPInfo{A version of this paper is submitted to Multiscale Modeling and Simulation: A SIAM Interdisciplinary Journal. Date of submission: 25.08.2016.}
}{}

\begin{document}

\title{\ourtitle\thanks{A version of this paper is submitted to Multiscale Modeling and Simulation: A SIAM Interdisciplinary Journal. Date of submission: 25.08.2016.}}

\author[1]{Isabel~Beckenbach}
\author[1]{Leon~Eifler}
\author[2]{Konstantin~Fackeldey}
\author[1]{Ambros~Gleixner}
\author[2]{Andreas~Grever}
\author[2]{Marcus~Weber}
\author[1]{Jakob~Witzig}

\affil[1]{Zuse Institute Berlin, Department Optimization, Takustr.~7, 14195~Berlin, Germany, \texttt{\{beckenbach,eifler,gleixner,witzig\}@zib.de}\medskip}
\affil[2]{Zuse Institute Berlin, Department Numerical Mathematics, Takustr.~7, 14195~Berlin, Germany, \texttt{\{fackeldey,grever,weber\}@zib.de}\medskip}

\ifthenelse{\zibreport = 1}{\zibtitlepage}{}

\maketitle

\begin{abstract}
  In this paper, we present a new, optimization-based method to exhibit cyclic
  behavior in non-reversible stochastic processes.  While our method is general,
  it is strongly motivated by discrete simulations of ordinary differential
  equations representing non-reversible biological processes, in particular
  molecular simulations.  Here, the discrete time steps of the simulation are
  often very small compared to the time scale of interest, i.e., of the whole
  process.  In this setting, the detection of a global cyclic behavior of the
  process becomes difficult because transitions between individual states may
  appear almost reversible on the small time scale of the simulation.

  We address this difficulty using a mixed-integer programming model that allows
  us to compute a cycle of clusters with maximum net flow, i.e., large forward
  and small backward probability. For a synthetic genetic regulatory network
  consisting of a ring-oscillator with three genes, we show that this approach can detect the most productive overall cycle,
  outperforming classical spectral analysis methods.  Our method applies to
  general non-equilibrium steady state systems such as catalytic reactions, for
  which the objective value computes the effectiveness of the catalyst.
\end{abstract}


\section{Introduction}

Simulation data stemming from chemical or biological processes typically lead to a huge amount of data points (time series) 
in some high dimensional space. Often a direct interpretation of these data for prediction or understanding of
the underlying physical process is almost impossible due to the range of spatial and temporal scales. 
This led to the development of coarse graining methods that provide relevant information of the system on a level with less complexity.  

One example widely used in the context of biological and chemical processes are Markov State Models (MSM)~\cite{BrownPandeNoe2014,DHFS_00,DW_05,PSK_11}.
In a MSM the underlying long time series is described by a Markov chain on some low dimensional space, \ie there exists a stochastic transition matrix $P$, whose entries $p_{ij}$ can be interpreted as the portion of the system that will transit from state~$i$ to state~$j$ in one time step.
If the vector~$v(t)$ represents a distribution at time step~$t$, the matrix vector 
multiplication $v(t)^TP=v(t+1)^T$ is a propagation of that distribution 
for one time step. The stationary distribution meets the condition $\pi^T=\pi^T P$, which means that $\pi$ is a steady state.

A Markov chain with $n$ states is called reversible if and only if the detailed balance condition
\begin{equation}\label{equ:detailedbalance}
 \pi_{i}p_{ij}=\pi_{j}p_{ji}
\end{equation}
holds for all $i,j=1,\ldots,n$. 
MSMs are well understood if they are applied to reversible processes, \eg simulation of a molecule in water, for which spectral clustering is a commonly used coarse graining method. 

In practice, however, many chemical and biological processes are not reversible.
If the detailed balance condition is not met for states $i$ and $j$, we can define a \emph{net flow} between these two states by taking the difference $\pi_{i}p_{ij}-\pi_{j}p_{ji}$.
It follows immediately that for $n>2$ we can find cycles of positive flow as long as the process is non-reversible. 

Non-reversible processes with a stationary distribution are sometimes called \emph{non-equilibrium steady state} (NESS) processes (see~\cite{DjurdjevacConradWeberSchuette2016}). A \emph{catalytic process} is one example of a NESS. 
A catalytic process is a chemical process in which the rate of chemical reactions is increased due to the presence of a catalyst. If one considers the (ensemble distribution of) states of the catalyst, they usually undergo a cycle in the conformational space, 
returning to its initial state at the end of the process. Thus,
the catalyst is in a NESS.
Each cycle transforms educts into products as illustrated in Figure~\ref{fig:catalytic-cycle}. The faster the productive cycle, the more effective the catalyst.
For processes with such a behavior we want to find a clustering that maximizes the net flow of the cycle.

Spectral clustering is applied in two cases. If the process consists of $m$ metastabilities, i.e.,
there is a set of states between which jumps only rarely occur, it is assumed that the spectrum of $P$ has $m$ leading eigenvalues (see~\cite{DHFS_00}). If the process consists of a dominant cycle of $m$ states, i.e., there is an ordered set of states and the process jumps to the next state with a high probability, it is assumed that the spectrum of $P$ has $m$ complex eigenvalues close to the unit circle (see~\cite{DjurdjevacConradWeberSchuette2016}).
Though processes involving different time scales like simulation of molecular dynamics are not likely to produce a dominant cycle. To analyze general cyclic behavior a more flexible approach is needed.  

\begin{figure}[ht]
  \centering
  \ifthenelse{\usecolors = 1}{
    \includegraphics[scale=0.4]{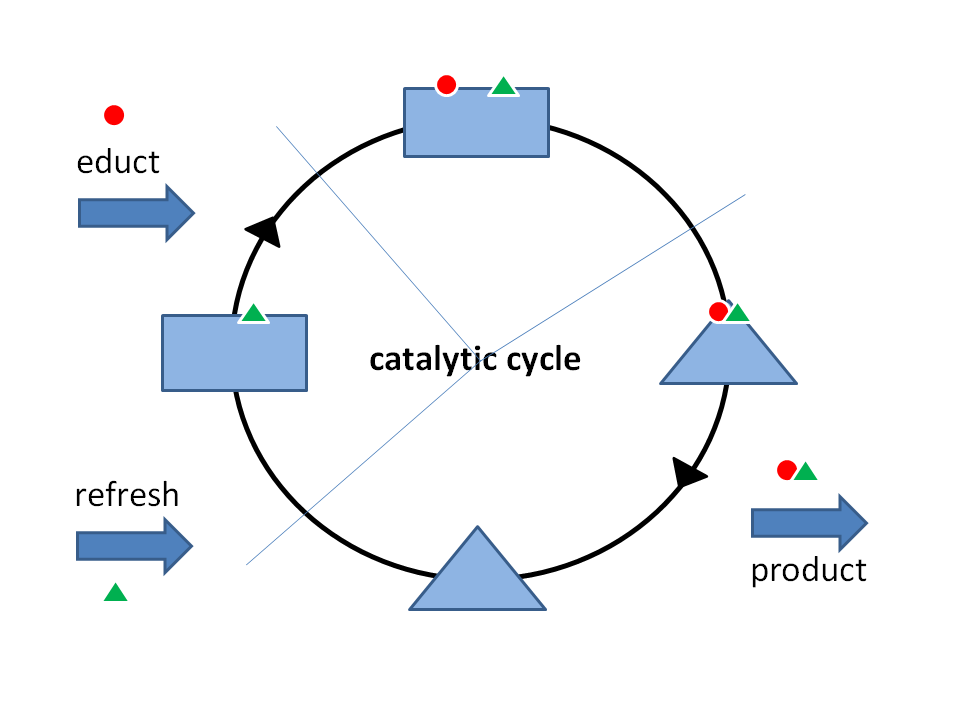}
  }{\includegraphics[scale=0.4]{cycle-sw.png}
  } 
  \caption{In a catalytic process the catalyst always returns to its initial state}
  \label{fig:catalytic-cycle}
\end{figure}

\paragraph{Contribution.}
Our main contribution is two-fold.
First, we develop a new technique to detect cycles in Markov processes called \emph{cycle clustering}.
It is general in the sense that the inherent non-reversibility may not be dominating, \ie the average ``speed'' of the cycle may be small relative to the overall time scale of interest.
In contrast, previous methods such as~\cite{DjurdjevacConradWeberSchuette2016} were designed for dominant cycles only.
The computational results show that we are able to find cycles that are far from being dominant, which makes our method applicable to biological processes like catalysis.
As such, our model could, for instance, be used to compute the effectiveness of the catalyst.

Second, as a more abstract methodological contribution, we model and solve the underlying clustering problem using mixed-integer optimization.
Compared to some classical methods, this approach requires few assumptions on the input data: transition matrices do not need to be reversible and there are no assumptions on its spectrum.
The separation of clustering model and the underlying solution algorithm helps to focus on the semantics of the clustering independently of how they can be computed.
The optimality guarantees given by mixed-integer programming solvers enable us to compute ``best'' clusterings and cycles with respect to the specified objective function.

The article is organized as follows. In Section~\ref{sec:background}, we give brief
introductions to spectral clustering and the current state-of-the-art methods
as well as an overview of mixed-integer programming techniques. In Section~\ref{sec:cycleclustering},
we develop our new, optimization-based method to exhibit cyclic behavior in
non-reversible Markov processes.
Moreover, we describe the underlying mixed-integer programming model and prove its complexity.
In Section~\ref{sec:results}, we investigate the viability of our approach by computational
experiments both on small-scale, artificially created instances and
the repressilator system from~\cite{ElowitzLeibler2000}.
In Section~\ref{sec:conclusion}, we give concluding remarks.

\section{Background}
\label{sec:background}

The clustering method presented in this paper is based on mixed-integer optimization.  As such it distinguishes itself from the spectral approach to cluster analysis most prevalent in molecular dynamics.
In the following, we will briefly describe the main ideas of spectral clustering and mixed-integer programming.

\subsection{Spectral Clustering}
\label{subsec:spectralclustering}

Spectral clustering is a common method to identify metastable sets in Markov State Models.
Given an undirected graph $G:=(V,E)$ defined by a set of $n$ vertices $V$, a set of edges $E \subseteq \binom{V}{2}$,
and a weighted adjacency matrix $W \in [0,1]^{n \times n}$, where $w_{ij} = w_{ji} > 0$ if and only if there exists an edge~$e_{ij} \in E$.
Moreover, let $D$ be a diagonal degree matrix, where $d_{ii}=\sum_{j}w_{ij} > 0$ for all $i \in V$.
The transition matrix $P$ that describes a random walk on $G$ is given by $P:=D^{-1}W$, \ie the probability to jump from node $i$ to $j$
is given by $\frac{w_{ij}}{d_{ii}}$. If the graph is connected and non-bipartite, there exists a unique stationary distribution vector $\pi \in [0,1]^n$, such that $\pi^{T}P=\pi^{T}$ (\cf~\cite{vonLuxburg2006}).

For each set of nodes $C \subseteq V$, the cut between $C$ and its complement $V \setminus C$ 
is defined by the set of edges $e_{ij} \in E$ with $i \in C$ and $j \in V \setminus C$. The weight of the cut between $C$ and $V \setminus C$ is defined by 
$\omega(C)=\sum_{i \in C, j \in V \setminus C} w_{ij}$. Now, the aim is to find a partition $C_{1},\ldots,C_{k}$ of $V$ 
such that the edges between clusters have small weight. A min-cut problem can be formulated as minimizing
\begin{align}
  \text{cut}(C_{1},\ldots,C_{k}):=\frac{1}{2}\sum_{i=1}^{k}\omega(C_{i},V \setminus C_{i}). \label{eq:mincut}
\end{align}
In practice, similar objective functions like RatioCut~\cite{hagen1992new} or Ncut~\cite{shi2000normalized} are used.
More general information on spectral clustering can be found in~\cite{vonLuxburg2006}.

A relaxation of the optimization problem leads to an eigenvalue problem for the first $k$ eigenvectors of a graph Laplacian $L=I-D^{-1}W=I-P$.
Eigenvectors corresponding 
to eigenvalues that are close to the Perron root $\lambda_1=1$
are the 
basis vectors of a special invariant subspace of $P$. There exists a 
transformed basis of this invariant subspace with the following 
property. The corresponding transformed basis vectors can be interpreted 
as the membership vectors of the metastable subsets of the state space. 
The method which finds this linear basis transformation is called PCCA+~\cite{DW_05}.

The analysis using a spectral approach is well understood when it comes to reversible Markov chains, \ie when the evolution of the process is invariant under time reversal. If the process is non-reversible then the weight matrix $W$ is not symmetric and the eigenvectors are not orthogonal. The canonical approach cannot be applied anymore. In~\cite{WeberFackeldey2015} a variation on the PCCA+ algorithm has been developed that can be applied to non-reversible matrices, called \emph{G-PCCA}. Instead of the spectral decomposition, G-PCCA uses the Schur decomposition such that the former follows as a special case.

Along a different line of research, some recent articles have tried to identify dominant structures such as cycles in non-reversible transition matrices, \eg~\cite{DjurdjevacConradBanischSchuette2014,DjurdjevacConradWeberSchuette2016}.
In these methods, the cycles are assumed to be ``dominant'', which means that there is a high probability inside the Markov chain to follow these cycles.
In this setting, complex pairs of eigenvalues of the transition matrices are clearly identifiable.

In constrast, the method presented in this paper aims at finding cycles in non-reversible transition matrices that are not necessarily dominant, but exhibit only rare circular jumps between metastable sets.
This type of transitions is typical for catalytic and other biological processes.
On the timescale of simulation the process is nearly reversible, but it has a small (in terms of probabilities) tendency towards the direction of a cyclic behavior. 
In these cases, non-reversibility is inherent to the process and does not stem from, \eg, truncation errors due to a finite sampling of a reversible Markov chain.
Hence, methods that try to make such input amenable to analytical tools for reversible Markov chains, \eg by computing the nearest reversible matrix as in~\cite{NielsenWeber2015}, are not applicable because they destroy the characteristics of the process.

\subsection{Mixed-Integer Programming}
A \emph{mixed-integer program} (MIP) is an optimization problem that can be written in the form
\begin{align*}
	(P) \qquad\qquad z_{MIP} = \min \{ c^Tx \;|\; Ax \geq b,\, \ell \leq x \leq u,\, x \in \Z^{l} \times \R^{n-l} \},
\end{align*}
with objective function $c \in \R^n$, constraint matrix $A \in \R^{m\times n}$, constraint right-hand side $b \in \R^m$,
and lower and upper bound vectors~$\ell, u \in (\R \cup \{\pm\infty\})^n$  on the variables.
When omitting the integrality conditions, we obtain the \emph{linear program} (LP)
\begin{align*}
	z_{LP} = \min \{ c^Tx \;|\; Ax \geq b,\, \ell \leq x \leq u,\, x \in \R^{n} \}.
\end{align*}
It constitutes a relaxation of the corresponding MIP and provides a lower bound on its optimum, \ie $z_{LP} \leq z_{MIP}$.
This fact plays an important role in the LP-based \bandb algorithm~\cite{dakin1965tree, land1960automatic}, the most widely used general algorithm to solve MIPs to global optimality.

LP-based \bandb is a divide-and-conquer method which starts by solving the LP relaxation of the problem to compute a lower bound and a solution candidate $x^\star$.
If $x^\star$ fulfills the integrality restrictions, the problem is solved to optimality; if not, it is split into (typically two) disjoint subproblems,
thereby removing $x^\star$ from the feasible region of both LPs. 
Typically, an integer variable $x_i$ with fractional solution value $x^\star_i$ is selected and the restrictions
$x_i \geq \ceil{x^{\star}_{i}}$ and $x_i \leq \floor{x^{\star}_{i}}$ are added to the two subproblems, respectively.
This step is called branching. As this process is iterated, we store and update the best solution $\tilde{x}$ found so far whenever one of the subproblems has an integral LP~solution.

The key observation is that a subproblem can be disregarded when its lower bound is greater or equal than the objective value of $\tilde{x}$. This is called bounding.
The branch-and-bound process is typically illustrated as a tree, \cf~Figure~\ref{fig:branchandbound}.
The root node represents the original problem and the two subproblems created by the branching step correspond to two child nodes being created for the current node.

In modern MIP solvers the general \bandb scheme is extended by various algorithms to enhance the performance, see, \eg~\cite{achterberg2008constraint, applegate1998solution, savelsbergh1994preprocessing} and many more.  Nevertheless, mixed-integer programming is complex both in theory ($\mathcal{NP}$-hard, see, \eg~\cite{GareyJohnson1979}) and in practice: As explained, state-of-the-art solvers eventually rely on enumerative search over an exponentially large solution space and may converge slowly.
However, even when terminated early for hard problem instances, they typically provide good solutions and give proven guarantees on the quality of the solutions returned.

\begin{figure}
    \centering
    \scalebox{.95}{
      \begin{tikzpicture}[
  treenode/.style = {shape=circle, rounded corners, draw, align=center, inner sep=.5, minimum width=2em},
  env/.style      = {treenode, font=\ttfamily\normalsize},
  level distance=1.25cm,sibling distance=2.5cm] 
  \tikzstyle{level 1}=[sibling distance=62mm] 
  \tikzstyle{level 2}=[sibling distance=28mm] 
  \tikzstyle{level 3}=[sibling distance=20mm] 
  \node[env]{$P$} 
    child{ 
      node [env]{$P_1$} 
      child {
	node [env]{$P_3$} 
	child {
	  node [env, dashed, thick]{$P_5$}
	  edge from parent node[pos=.3, fill=white, inner sep=0, anchor=east] {$x_{i_3} \leq \floor{x^{\star, P_5}_{i_3}}$}
	} 
	child{ 
	  node [env, dotted, thick]{$P_6$} 
	  edge from parent node[pos=.3, fill=white, inner sep=0, anchor=west] {$x_{i_3} \geq \ceil{x^{\star, P_5}_{i_3}}$}
	}
	edge from parent node[pos=.3, fill=white, inner sep=0, anchor=east] {$x_{i_2} \leq \floor{x^{\star, P_1}_{i_2}}$}
      } 
      child{ 
	node [env, fill=gray!30]{$P_4$}
	edge from parent node[pos=.3, fill=white, inner sep=0, anchor=west] {$x_{i_2} \geq \ceil{x^{\star, P_1}_{i_2}}$}
      }
      edge from parent node[pos=.5, fill=white, inner sep=0] {$x_{i_1} \leq \floor{x^{\star, P}_{i_1}}$}
    }
    child{ 
      node [env]{$P_2$} 
      child {
	node [env, fill=gray!30]{$P_7$} 
	edge from parent node[pos=.3, fill=white, inner sep=0, anchor=east] {$x_{i_4} \leq \floor{x^{\star, P_2}_{i_4}}$}
      } 
      child{ 
	node [env, fill=gray!30]{$P_8$} 
	edge from parent node[pos=.3, fill=white, inner sep=0, anchor=west] {$x_{i_4} \geq \ceil{x^{\star, P_2}_{i_4}}$}
      } 
      edge from parent node[pos=.5, fill=white, inner sep=0] {$x_{i_1} \geq \ceil{x^{\star, P}_{i_1}}$}
  };
\end{tikzpicture}
    }
    \caption{Illustration of a \bandb tree. Solving the LP relaxation of $P$, $P_1$, and $P_3$ led to branching steps on variables $i_1$, $i_2$, and $i_3$, respectively. 
    First feasible solution $\tilde{x}$ obtained by solving the LP relaxation of $P_5$ (dashed node). Subproblem $P_6$ has an infeasible LP relaxation after branching on $i_3$ (dotted node).
    Afterwards, subproblems $P_4$, $P_7$, and $P_8$ can be disregarded due to bounding (gray shaped), \ie the lower bound of each subproblem is not smaller than the objective value of $\tilde{x}$. 
    Since all subproblems are processed, $\tilde{x}$ is an optimal solution of $P$.}
    \label{fig:branchandbound}
\end{figure}
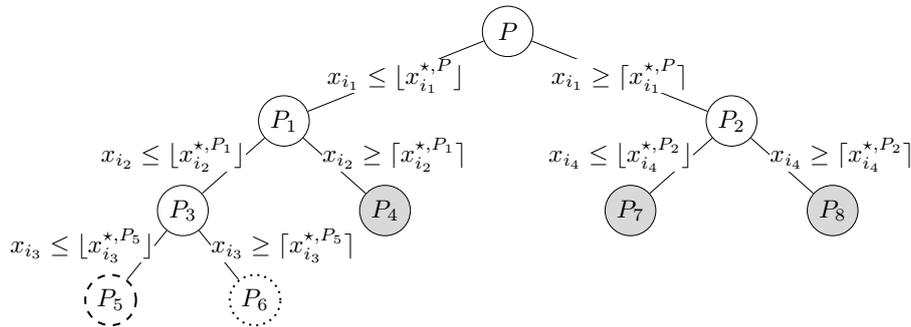

\section{Cycle Clustering}
\label{sec:cycleclustering}

In this section we present a new method for detecting global cyclic behavior of a Markov process by partitioning its state space into ordered clusters.
The technique, which we call \emph{cycle clustering}, can be applied to any discrete-time Markov process over a discrete (or discretized) state space with a stationary distribution.
We show how an ``optimal'' cycle clustering can be computed using algorithms from mixed-integer programming.

\subsection{Setting}

We consider a Markov process with a finite set of states~$\Bins = \{1,\ldots,n\}$.
We will call states also \emph{bins} in order to indicate that they might stem from a discretization of a continuous space as it is common in molecular simulations.
However, this is not an underlying assumption of our method.

Let~$\TransMat \in [0,1]^{n \times n}$ be the matrix of conditional transition probabilities, where $p_{ij}$ is equal to the probability of moving from bin~$i$ to bin~$j$ in one time step when initially in bin~$i$.
Then the only requirements for applying our method are
\begin{enumerate}
\item that~$P$ is right stochastic, \ie all row sums are one, and
\item that~$P$ has a stationary distribution~$\pi$, \ie $\pi \in [0,1]^{n}$ is a left-hand eigenvector for eigenvalue~1, $\pi^T \TransMat = \pi^T$, and~$\sum_{i=1,\ldots,n} \pi_i = 1$.
\end{enumerate}
Using the stationary distribution, we can compute the matrix~$W = \diag(\pi) P$ of \emph{unconditional} transition probabilities, which is the main input data for our method.
The entry~$q_{ij} = \pi_i p_{ij}$ equals the probability or intense of transitions from~$i$ to~$j$ in the whole ensemble of transitions.

\subsection{Net Flow and Coherence}

While in many biological processes such as catalysis it may be intuitively clear how to define cyclic behavior using application-specific interpretations of the state space, it is much less obvious in the abstract formalism of a Markov state model.
Because of the probabilistic aspect it is not meaningful to look for a sequence of states that are visited in one fixed cyclical order.
Moreover, a trajectory may not even visit all states as it completes one ``iteration'' of the cycle.

We address these difficulties by combining cycle detection with clustering.
Our goal is to partition the set of states into a fixed number of clusters and order them in form of a cycle such that with high probability we will encounter
\begin{itemize}
\item transitions from one cluster to the next cluster in cycle direction, or
\item transitions within one cluster, but
\item no transitions between clusters in backward direction.
\end{itemize}
To quantify this, we introduce the following measure of non-reversibility between two sets of states.

\begin{Def}[net flow]
  Given two disjoint sets of states $A, B \subseteq \Bins$, $A \cap B = \emptyset$, we call
  \begin{align*}
    \NetFlow(A,B) := \sum_{i \in A, j \in B} (\pi_{i} p_{ij}-\pi_{j}p_{ji}) = \sum_{i \in A, j \in B} (q_{ij}-q_{ji})
  \end{align*}
  the \emph{net flow} from set~$A$ to set~$B$.
\end{Def}

The net flow~$\NetFlow(A,B)$ corresponds to the portion of particles transiting from~$A$ to~$B$ minus the portion transiting backwards from~$B$ to~$A$ in one time step.
Its value equals the sum of the deviations from the detailed balance conditions~\eqref{equ:detailedbalance}.
The net flow is signed and by definition, $\NetFlow(A,B) = -\NetFlow(B,A)$ for all~$A, B$.

To ensure that, at the same time, clusters contain related groups of states, we use the following definition.

\begin{Def}[coherence]
  Given a set of states $A \subseteq \Bins$, we call
  \begin{align*}
    \Coherence(A) := \sum_{i,j \in A} \pi_{i} p_{ij} = \sum_{i,j \in A} q_{ij} 
  \end{align*}
  the \emph{coherence} of set~$A$.
\end{Def}

The coherence~$\Coherence(A)$ is hence equal to the unconditional probability of residing and remaining within~$A$ given the stationary distribution~$\pi$.
As such, it can be interpreted as a proxy for measuring closeness in the original state space.

\subsection{Clustering Model}
\label{subsec:clusteringmodel}

By an \emph{$m$-cycle clustering} we denote a partitioning of the set of states into $m$~pairwise disjoint clusters,
\begin{equation*}
  \Bins = \bigcup_{k=1}^m \Cluster_k,\, \Cluster_k \cap \Cluster_\ell = \emptyset \text{ for all } k\not=\ell,
\end{equation*}
endowed with the cyclic order~$\Cluster_1 \rightarrow \Cluster_2 \rightarrow \ldots \rightarrow \Cluster_m \rightarrow \Cluster_1$.
We are interested in cycle clusterings with large total net flow between consecutive clusters,
\begin{equation*}
  \sum_{k=1}^{m-1} \NetFlow(\Cluster_k,\Cluster_{k+1}) + \NetFlow(\Cluster_m,\Cluster_1),
\end{equation*}
and large total coherence of the individual clusters, \ie
\begin{equation*}
  \sum_{k=1}^{m} \Coherence({\Cluster_k}).
\end{equation*}

To use matrix notation, we can encode a clustering in an assignment matrix~$X \in \{0,1\}^{n \times m}$,
\begin{equation*}
  X_{ik}=\begin{cases} 1 & \text{for } i \in \Cluster_k, \\ 0 & \text{otherwise.} \end{cases}
\end{equation*}
Then our objective can be expressed in terms of the projected matrix of unconditional transition probabilities
\begin{equation*}
  \overline{W} := X^T W X \in [0,1]^{m \times m}.
\end{equation*}
Its diagonal entries~$\overline{w}_{kk}$ carry the coherences~$\Coherence(\Cluster_k)$.  The net flow values~$\NetFlow(\Cluster_k,\Cluster_\ell)$ equal~$\overline{w}_{k\ell} - \overline{w}_{\ell k}$, \ie they can be read from the off diagonal entries of
\begin{equation*}
  \Delta := \overline{W} - \overline{W}^T = X^T (W - W^T) X.
\end{equation*}
By construction, the diagonal entries of~$\Delta$ are zero.  Furthermore, from~$\pi^T P = \pi^T$ it follows that the row sum vector of~$W$ equals its column sum vector, \ie the row sums and column sums of~$\Delta$ are all zero.

As a first consequence, $\Delta$ is the zero matrix if $m \leq 2$, \ie a clustering into a 2-cycle cannot exhibit any non-reversibility.
In this sense, the smallest interesting case is a \emph{3-cycle clustering}.
As a second consequence, for any 3-cycle clustering, $\Delta$ must have the special structure
\begin{equation*}
  \Delta = \begin{pmatrix}
    0 & \varepsilon & -\varepsilon \\
    -\varepsilon & 0 & \varepsilon \\
    \varepsilon & -\varepsilon & 0 
  \end{pmatrix},\label{eq:casem3}
\end{equation*}
where we may assume~$\varepsilon \geq 0$ after reordering.
Hence, in a 3-cycle clustering, the net flow between each two clusters, $\varepsilon$, is identical.
Maximizing the total net flow~$3\varepsilon$ is equivalent to maximizing pairwise non-reversibility.
In our experiments later, we will also focus on this prototypical case of a cycle clustering with~$m=3$.

Finally, note that the two objectives, non-reversibility in terms of net flow and coherence, are not necessarily aligned:
In general, we cannot assume that there will be one clustering that maximizes both criteria at the same time.
Hence, we combine both and use a scaling parameter~$\alpha > 0$ to control the emphasis on coherence to obtain the weighted objective
\begin{equation}\label{equ:objective}
  \sum_{k=1}^{m-1} \NetFlow(\Cluster_k,\Cluster_{k+1}) + \NetFlow(\Cluster_m,\Cluster_1)
  + \alpha \sum_{k=1}^{m} \Coherence({\Cluster_k}).
\end{equation}
In our experiments, we will use a small value of~$\alpha = 0.001$ as a default.

The following gives an example where both criteria are important to detect a meaningful 3-cycle.



\begin{example}\label{ex:cycleclustering}
  Consider the process shown in Figure~\ref{fig:cycleclustering}, which has nine bins grouped into three ``natural'' clusters.
  The unconditional transition probabilities are attached to the edges, scaled by a factor of nine, the number of bins, for better readability.
  The example is constructed symmetrically such that there is a small cyclic flow between bins~$1$, $2$, and $3$.
  Each of these bins is connected to two more bins in a reversible fashion.

  To obtain the highest possible total net flow of~$0.3/9$, a 3-cycle clustering must assign the states~$1$, $2$, and $3$ to different clusters.
  The assignment of the remaining states has no influence on the net flow and could therefore be performed arbitrarily---unless we take into account coherence.
  Coherence is maximized by clustering~$4$ and~$5$ with~$1$, $6$ and~$7$ with~$2$, and $8$ and~$9$ with~$3$, thus detecting the natural cyclic structure of the process.
\qed
\end{example}

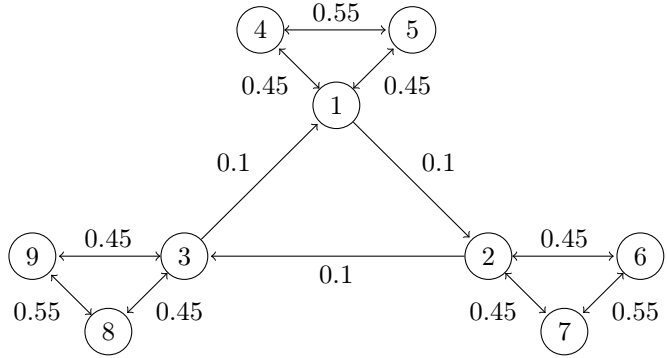
\begin{figure}
  \centering
  \begin{tikzpicture}[->,shorten >=1pt,auto,
onenode/.style = {shape=circle, rounded corners,
                     draw, align=center
                     },
twonode/.style = {shape=circle, rounded corners,
                    	draw, align=center
                    	},
threenode/.style = {shape=circle, rounded corners,
                    	draw, align=center}
                    	]
                    
      \foreach \place/\name in { {(2,2)/1}, {(1,3)/4}, {(3,3)/5}}
      \node[onenode] (\name) at \place {\name};
      \foreach \place/\name in { {(4,0)/2}, {(6,0)/6}, {(5,-1)/7}}
            \node[twonode] (\name) at \place {\name};
      \foreach \place/\name in {{(0,0)/3},{(-2,0)/9}, {(-1,-1)/8}}
                  \node[threenode] (\name) at \place {\name};
     \path[->]           (1)  edge 	     node {0.1} (2);
     \path[<->]          (1)  edge    node {0.45}  (4);
     \path[<->]          (1)  edge    node [below right]  {0.45} (5);
     \path[<->]          (4)  edge    node {0.55} (5);
     \path[->]           (2)  edge 	     node {0.1} (3);
     \path[<->]          (2)  edge    node {0.45}  (6);
     \path[<->]          (2)  edge    node [below left]  {0.45} (7);
     \path[<->]          (6)  edge    node {0.55} (7);
     \path[->]           (3)  edge 	     node {0.1} (1);
     \path[<->]          (3)  edge    node {0.45}  (8);
     \path[<->]          (3)  edge    node [above] {0.45}  (9);
     \path[<->]          (8)  edge    node {0.55} (9);
\end{tikzpicture}
  \caption{A cyclic process with non-dominant cycle. The edges weights are the unconditional transition probabilities $\pi_i p_{ij}$. 
  For better readability they have been scaled by a factor of nine.}
  \label{fig:cycleclustering}
\end{figure}

\begin{Rem}
  In earlier models, we had also experimented with the natural idea of maximizing some measure of reversibility within each cluster.
  However, this incurs difficulties when the backward and forward probability between two states is zero or very small.
  First, this introduces a connection between two states that might not be at all or are only very weakly related.
  Second, it is inherently instable to classify such edges as irreversible or reversible because of the small differences.
  Coherence, as used here, avoids all these disadvantages.
\end{Rem}

\subsection{Mixed-Integer Programming Formulation}
\label{subsec:MIP}
In order to actually compute optimal cycle clusterings w.r.t.~\eqref{equ:objective}, we use the following MIP model.
For each bin $i \in \Bins = \{1,\ldots,n\}$ and cluster $k \in \Clusters = \{1,\ldots,m\}$, we introduce a binary decision variables $x_{ik}$ with 
\begin{align*}
  x_{ik} = 1 \Longleftrightarrow i\in \Cluster_k \Longleftrightarrow \text{ bin } i \text{ is assigned to cluster } k,
\end{align*}
which correspond to the entries of the assignment matrix~$X$ used before.
Furthermore, with slight abuse of notation, we introduce continuous variables~$\NetFlow_k \in \R_{\geq 0}$ for the net flow from cluster $k$ to~$\phi(k)$, where 
\begin{align*}
  \phi\;:\;\Clusters\mapsto\Clusters,\;\phi(k) = \begin{cases} k+1 &\text{ if } k < m, \\ 1 &\text{ otherwise,}\end{cases}
\end{align*}
and continuous variables $\coh_k \in \R_{\geq 0}$ for the coherence of cluster $k \in \Clusters$.
Then computing an optimal cycle clustering for a fixed number of $m$~clusters can be expressed as the following MIP model:
 \begin{align}
  \lefteqn{\hspace*{-5.3em} \max\; \sum_{k \in \Clusters}\NetFlow_k + \alpha \cdot \sum_{k \in \Clusters} \coh_k} \\
  \text{s.t.} \hspace*{1.8ex} \sum_{k \in \Clusters} x_{ik} &= 1 && \fa i \in \Bins \label{eq:assignment} \\
  \sum_{i \in \Bins} x_{ik} &\ge 1 && \fa k \in \Clusters \label{eq:setcover} \\
  \NetFlow_k &= {\sum_{i,j \in \Bins} q_{ij} (x_{ik} x_{j\phi(k)} -	 x_{i\phi(k)} x_{jk})} && \fa k \in \Clusters \label{eq:irreversibility} \\
  \coh_k &=  \sum_{i,j \in \Bins} q_{ij} x_{ik} x_{jk} &&  \fa k \in \Clusters \label{eq:coherence} \\
  x_{ik} &\in \{0,1\} && \fa i \in \Bins, k \in \Clusters \\
  \NetFlow_k, \coh_k &\in \R_{\geq 0} && \fa k \in \Clusters
\end{align}

Constraints of type~\eqref{eq:assignment} ensure that each bin $i$ is assigned to exactly one cluster $k$. Constraints~\eqref{eq:setcover} assert that there are no empty clusters.
The net flow between two consecutive clusters~$k$ and $\phi(k)$ is described by constraints of type~\eqref{eq:irreversibility}.
The coherence within each cluster $k$ is modeled by constraints of type~\eqref{eq:coherence}.

The products of binary variables appearing in constraints~\eqref{eq:irreversibility} and~\eqref{eq:coherence} are nonlinear.
We have applied a standard reformulation technique~\cite{fortet1960algebre} to obtain a mixed-integer \emph{linear} programming formulation that can be solved by standard state-of-the-art MIP solvers.
This requires the introduction of additional auxiliary variables, but yields significantly lower solution times than using global mixed-integer nonlinear programming solvers.

\begin{Rem}
  We want to point out that the MIP model is even more general than our initial development of the cycle clustering approach.
  It only requires a non-negative matrix~$W$ as input and does not rely on the form~$W = \diag(\pi) P$ with~$P$ being a stochastic matrix and~$\pi$ its stationary distribution vector.
\end{Rem}

\subsection{Complexity of Cycle Clustering}
\label{subsec:complexity}


While it is known that mixed-integer programming is $\NP$-hard in general~\cite{schrijver2003combinatorial}, special subclasses still may be easier.
In this section we discuss the complexity of cycle clustering and show that it is $\NP$-hard by a reduction from the multiway cut problem~\cite{dahlhaus1992complexity}.

\begin{Def}[multiway cut]
	Let $G=(V,E)$ be a graph with non-negative edge weights $c(e) \ge 0$ and a set of specified vertices $S = \{s_1,\ldots,s_m\} \subseteq V$ called terminals. A \emph{multiway cut} is a subset of edges $E' \subseteq E$ that separates the terminals $s_1,\ldots, s_m$ in the sense that there exists no path from any terminal to any other terminal in $(V,E \setminus E')$. The multiway cut problem is that of finding a weight-minimal multiway cut.
\end{Def}

The multiway cut problem is $\NP$-hard for any fixed~$m \geq 3$~\cite{dahlhaus1992complexity}, which allows us to prove the following.

\begin{Thm}\label{thm:complexity}
  The cycle clustering problem as defined in Section~\ref{subsec:clusteringmodel} is $\NP$-hard for any~$m \in \N$, $m \geq 3$, and any~$\alpha > 0$.
\end{Thm}

\begin{proof}
  Suppose we are given a multiway cut instance over an undirected graph $G=(V,E)$ with the set of nodes~$V = \{1,\ldots,n\}$, the set of edges~$E \subseteq \binom{V}{2}$, and edge weights~$c(e) \ge 0$.
  Let $S = \{s_1, \ldots, s_m\}$ be the set of terminals for~$m \geq 3$.
  Because edges between terminals contribute a constant offset to the objective value of any multiway cut, we may assume w.l.o.g.\ that~$E \cap  \binom{S}{2} = \emptyset$.
  Because disconnected non-terminal nodes can be assigned arbitrarily, we may assume w.l.o.g.\ that ~$\sum_{v\in V\colon \{uv\}\in E} c(\{uv\}) > 0$ for all~$u \in V \setminus S$.

  In order to show that this instance can be solved as a cycle clustering problem, we construct a directed graph~$D = (V, A)$ with arc set
  \begin{align*}
    A &= \{ (uv) \in V\times V\mid \{uv\} \in E,\ u < v \} \tag{forward arcs} \\
     &\hspace{1.2pt}\cup\hspace{1.2pt} \{ (uv) \in V\times V\mid \{uv\} \in E,\ u > v \} \tag{backward arcs} \\
     &\hspace{1.2pt}\cup\hspace{1.2pt} \{ (s_{i}s_{\phi(i)}) \in S\times S\mid \{s_{i}s_{\phi(i)}\} \notin E \}. \tag{auxiliary arcs}
  \end{align*}
  The first two parts split edges into forward and backward arcs and the third set, which may be empty, ensures that the cycle~$s_1 \rightarrow s_2 \rightarrow \ldots \rightarrow s_m \rightarrow s_1$ is present in~$D$.

  To define arc weights, we partition the edges of~$G$ into
  \begin{align*}
    \check{E} = \{ \{uv\} \in E \mid \vert \{u,v\} \cap S \vert \leq 1\} \quad \text{and} \quad
    \hat{E} = \{ \{uv\} \in E \mid u,v \in S \}
  \end{align*}
  and set
  \begin{align*}
    d(a) = 
    \begin{cases}
    M      & \fa a = (s_{i}s_{\phi(i)}) \in A, \\
    0      & \fa a = (s_{i}s_j) \in A \text{ with } j \neq \phi(i), \\
    c(e)/2 & \fa a = (uv) \in A, \{uv\} \in \check{E}.
    \end{cases}
  \end{align*}
  Here $M$ is chosen sufficiently large such that each terminal will later be forced into a different cluster, \ie
  \begin{align*}
    M > \sum_{e \in E} \alpha c(e) \geq \sum_{e \in \check{E}} \alpha c(e). 
  \end{align*}
  From this construction we derive a weighted adjacency matrix $Q \in\R^{n \times n}$ via
  \begin{align*}
    q_{uv} = 
    \begin{cases}
      d(a) & \text{if } a = (uv) \in A, \\
      0    & \text{otherwise}.
    \end{cases}
  \end{align*}
  Normalization of each row by its row sum~$\Norm{Q_{u\cdot}}_{1} = \sum_{u^\prime \in V} q_{uu^\prime}$, which by assumption is non-zero, gives a stochastic matrix $P\in [0,1]^{n \times n}$ with entries
  \begin{align*}
    p_{uv} = \frac{q_{uv}}{\Norm{Q_{u\cdot}}_{1}}.
  \end{align*}
  This transition matrix~$P$ has a unique stationary distribution~$\pi$ given by
  \begin{align*}
    \pi_{u} = \frac{\Norm{Q_{u\cdot}}_{1}}{\sum_{u' \in V} \Norm{Q_{u'\cdot}}_{1}}
  \end{align*}
  for~$u = 1,\ldots,n$.
  The corresponding matrix~$W$ of unconditional transition probabilities has entries
  \begin{align*}
    w_{uv}
    = \pi_u p_{uv}
    = \frac{\Norm{Q_{u\cdot}}_{1}}{\sum_{u' \in V} \Norm{Q_{u'\cdot}}_{1}} \cdot \frac{q_{uv}}{\Norm{Q_{u\cdot}}_{1}}
    = \frac{q_{uv}}{\sum_{u^\prime,v^\prime \in V} q_{u^ \prime v^\prime}}
  \end{align*}
  for~$u,v = 1,\ldots,n$, and thus
  \begin{align*}
    w_{uv} - w_{vu}
    = \begin{cases}
      M & \text{for } u = s_i, v = s_{\phi(i)},\\
      -M & \text{for } v = s_{\phi(i)}, u = s_i,\\
      0 & \text{otherwise.}
    \end{cases}
  \end{align*}
  Hence, only edges between consecutive terminals violate the detailed balance condition and contribute to the net flow.
  
  Now let~$\Cluster_1,\ldots,\Cluster_m$ be an optimal solution of the cycle clustering problem w.r.t.\ the constructed matrix~$P$ and the stationary distribution~$\pi$.
  W.l.o.g., assume~$s_1 \in \Cluster_1$.
  Due to the choice of~$M$, in any optimal clustering, terminal~$s_k$ must be in cluster~$\Cluster_k$.

  Finally, since the assignment of non-terminal nodes does not affect the net flow, they must be assigned such as to maximize coherence.
  The following calculation shows that maximizing this remaining part of the objective function is equivalent to minimizing the weight of the edges in the corresponding multiway cut,
  \begin{align*}
    \alpha \sum_{k=1}^{m} \Coherence({\Cluster_k})
    &= \alpha \sum_{k=1}^{m} \sum_{u,v \in \Cluster_k} w_{uv}
    = \alpha \sum_{k=1}^{m} \sum_{u,v \in \Cluster_k} \frac{q_{uv}}{\sum_{u^\prime,v^\prime \in V} q_{u^\prime v^\prime}}\\
    &= \frac{\alpha}{\sum_{u^\prime,v^\prime \in V} q_{u^\prime v^\prime}} \sum_{k=1}^{m} \sum_{u,v \in \Cluster_k} q_{uv}
    = \frac{\alpha}{\sum_{u^\prime,v^\prime \in V} q_{u^\prime v^\prime}} \sum_{k=1}^{m} \sum_{e \in E \cap \binom{\Cluster_k}{2}} \hspace*{-1.5ex}c(e)\\
    &= \underbrace{\frac{\alpha}{\sum_{u^\prime,v^\prime \in V} q_{u^\prime v^\prime}}}_{\text{constant } >0} \Big( \underbrace{\sum_{e \in E} c(e)}_{\text{constant}} - \hspace*{-1.5ex}\underbrace{\sum_{e \in E \setminus \bigcup_{k = 1}^{m}\binom{\Cluster_k}{2}} c(e)}_{\text{multiway cut weight}} \hspace*{-.5ex}\Big)
  \end{align*}

  To summarize, we gave a polynomial reduction of the multiway cut problem to cycle clustering, proving that cycle clustering is $\NP$-hard.
\end{proof}

\section{Computational Experiments}
\label{sec:results}

To evaluate our new clustering approach we used both synthetic instances and a well-known system of differential equations that model the interaction of genes~\cite{ElowitzLeibler2000}.
In the first part of this section we describe the set of instances we have used for the computational experiments in more detail.
In the second part we discuss the solving environment and the software we have used. Finally, we present our computational results.

\subsection{Testset}
\label{subsec:testset}

\paragraph{Catalytic Cycle.}

To create data sets that resemble molecular dynamical simulations, a hybrid Monte-Carlo method (HMC) was applied to a synthetic, two-dimensional energy landscape~$\Omega$ as described in~\cite{Brooks2011,FackeldeyWeber2014}. 
In this variant of an HMC, the system is propagated with a drift to one of the minima defined by the potential~$\Omega$ and an additional random value, followed by a Metropolis-like acceptance step that assures the convergence of the distribution defined by the function~$\Omega$.
If the system enters a predefined set the drift is updated for the next state of the cycle.
This creates the dynamics of a metastable system with rare asymmetric jumps, that one would expect from a catalytic cycle.

\begin{algorithm}[t]
\caption{HMC with drift}
\label{alg:HMC}
\begin{algorithmic}[1]
\STATE \algorithmicrequire { start vector $x_{0}$, inverse temperature $\beta$, $N$, drift $d$, random vectors~$r_{1},\ldots,r_{N}$}, uniformly distributed numbers~$u_{1},\ldots,u_{N}\in [0,1]$
\STATE \algorithmicensure { trajectory $x_{0},\ldots,x_{N-1}$  }
\FOR{$i=1$ \TO $N-1$}
\STATE $x_{new} \gets x_{i-1}+r_{i}+d$
\IF {$\exp(-\beta (\Omega(x_{new})-\Omega(x_{i-1})))<u_{i}$}
\STATE $x_{i} \gets x_{new}$
\STATE $d\,\, \gets \text{update}(d)$
\ELSE \STATE $x_{i} \gets x_{i-1}$
\ENDIF
\ENDFOR
\end{algorithmic}
\end{algorithm}
Algorithm~\ref{alg:HMC} yields a sampling of length $N$, which is used to compute the transition matrix. 
We extracted~$n$ vectors~$c_{1},\ldots,c_{n}$ from the sampling such that the fill distance
\[
h:=\max_{j=1,\ldots,N} \min_{i=1,\ldots,n} \|x_{j}-c_{i}\|_{2}
\]
was minimized. The vectors $c_{i}$ are the centers of regions (bins) and following~\cite{Weber2011}, radial basis functions
\begin{align*}
\Phi_{i}(x):=\frac{\exp(-\|x-c_{i}\|_{2}^{2})}{\sum_{k=1}^{n}\exp(- \|x-c_{k}\|_{2}^{2})}
\end{align*}
with values in $ (0,1) $ were used as membership function, \ie instead of assigning every $x$ to one bin, it is assigned to bin $i$ with the fraction $\Phi_{i}(x)$.
The transition matrix $P$ was then defined as
\[
p_{ij}:= \frac{ \sum_{k=0}^{N} \Phi_{i}(x_{k}) \Phi_{j}(\tilde{x}_{k})} {\sum_{k=0}^{N} \Phi_{i}(x_{k})}
\]
where
the notation $\tilde{x}_{k}$ refers to a propagation of the system by a time step~$t$, \ie $\tilde{x}_{k}$ is the state of $x_{k}$ after $t$ steps. In the examples we set $t=1$.

The method was applied to examples with three, four, and six minima to show the influence of the applied drift on the clustering and the value of maximal flow.
The potential functions have their minima arranged on a circle around a single maximum and are of the form

\begin{align*}
\Omega_{3}(x,y) =\; &6\exp\left[-3(x^{2}+y^2)\right]  \\ 
-&8\exp\left[-(x-x^{\star})^2-(y+y^{\star})^2\right]-8\exp\left[-(x+x^{\star})^2-(y+y^{\star})^2\right]\\
-&8\exp\left[-x^2-(y-1)^2\right]\\
\Omega_{4}(x,y) =\; &4\exp\left[-3(x^{2}+y^2)\right]  \\ 
-&8\exp\left[-x^2-(y-1.5)^2\right]-8\exp\left[-(x-1)^2-y^2\right] \\
-&8\exp\left[-(x+1)^2-y^2\right]-8\exp\left[-x^2-(y+1.5)^2\right]\\
\Omega_{6}(x,y) =\; &4\exp\left[-3(x^{2}+y^2)\right]  \\ 
-&8\exp\left[-(x-2x^{\star})^2-(y+2y^{\star})^2\right]-8\exp\left[-(x+2x^{\star})^2-(y+2y^{\star})^2\right]\\
-&8\exp\left[-(x+2x^{\star})^2-(y-2y^{\star})^2\right]-8\exp\left[-(x-2x^{\star})^2-(y-2y^{\star})^2\right]\\
-&8\exp\left[-(x+2)^2-(y-1)^2\right]-8\exp\left[-(x-2)^2-y^2\right]    
\end{align*}
with $x^{\star}=0.5$ and $y^{\star}=0.5\sqrt{3}$.
The minima resemble the metastable macroscopic states of the system.
Trajectories with $N=10000$ steps were simulated to construct matrices with $n=20$ bins.

\paragraph{Repressilator.}

The repressilator is a system of differential equations that describes a synthetic genetic regulatory network~\cite{ElowitzLeibler2000} consisting of three genes, TetR, $\gamma$cI, and LacI. Each of the genes produces a protein $p$ that represses the production of mRNA $m$ of another gene. The symmetric system was described in \cite{ElowitzLeibler2000} by the equations
\begin{align*}
\dfrac{dm_{A}}{dt}=-m_{A}+\dfrac{v}{1+p_{C}^{h}}+v_{0} && \dfrac{dp_{A}}{dt}=-\beta(p_{A}-m_{A}) \\
\dfrac{dm_{B}}{dt}=-m_{B}+\dfrac{v}{1+p_{A}^{h}}+v_{0} && \dfrac{dp_{B}}{dt}=-\beta(p_{B}-m_{B}) \\
\dfrac{dm_{C}}{dt}=-m_{C}+\dfrac{v}{1+p_{B}^{h}}+v_{0} && \dfrac{dp_{C}}{dt}=-\beta(p_{C}-m_{C})
\end{align*} 
with $v$ = 298.2 transcriptions per second, $\beta$ = 1/5 the ratio of protein decay
rate to mRNA decay rate, a growth constant~$v_{0}$ = 0.03, and a Hill coefficient~$h = 2$.

Trajectories were started at 200 points in the six-dimensional cube of the interval $[0,20]$ and simulated for 1.5 seconds using the ode45 function by Matlab.
The starting points were generated by a Niederreiter sequence (see \cite{LN_84,N_88}) with values from [0,1] and scaled afterwards.
The transition matrix was defined similar as above by
\[
p_{ij}=\dfrac{\exp(-0.2\|x^{start}_{i}-x^{end}_{j}\|)}{\sum_{k=1}^{200}\exp(-0.2\|x^{start}_{i}-x^{end}_{k}\|)}, \quad i,j=1,\ldots,200.
\]

\subsection{Testing Environment}
\label{subsec:environment}

For computing optimal cycle clusterings we used the mixed-integer programming solver SCIP, which is free for academic purposes~\cite{GamrathFischerGallyetal.2016}.

Although in theory SCIP could solve general MIP formulations out-of-the-box, our cycle clustering MIPs proved hard in practice.
Hence, to speed up the solving process of SCIP, we implemented three problem-specific heuristics that try to find good primal solutions.  They are called at the root node and during the branch-and-bound search:
\begin{itemize}
\item First, we implemented a greedy heuristic to construct a feasible clustering by iteratively assigning the bins to clusters.
The heuristic starts with the assignment of bin $1$ to cluster $1$ and assigns all remaining bins iteratively.
Therefore, the best possible assignment w.r.t.\ non-reversibility and coherence is computed in each assignment step.
\item Second, inspired by the approach with Schur vectors~\cite{WeberFackeldey2015}, we use the solution of the LP relaxation at each
node within the \bandb procedure as a starting point for a rounding heuristic that knows the specific problem structure.
\item Third, we implemented an improvement heuristic similar to~\cite{KernighanLin1970} that iteratively tries to identify a bin that can be moved to a different cluster such that the objective function is increased.
\end{itemize}
Note that none of these procedures is guaranteed to find an optimal clustering or even to succeed at all at finding a feasible solution.
However, applied regularly as part of the global solution process of SCIP they help to accelerate the convergence of the primal and dual bound significantly.

All tests were run sequentially on identical machines with an Intel Xeon Quad-core with 3.2\,GHz and 48\,GB of RAM.
To balance net flow and coherence in the objective function, we used a value of~$\alpha = 0.001$.

\subsection{Computational Results}

\begin{table}[t]
\centering
\setlength{\tabcolsep}{4pt}
\begin{tabular}{lcr*{4}{r}}
\toprule
instances  	& drift			& best obj. 			& obj.\ bound                        & net flow 				& coherence & time [s] \\
\midrule                                                                                         
$\Omega_{3}$			& 0.1	 		& 0.0018 	& 0.0018 & 0.0013 & $0.5374$     & 0.1 	\\
$\Omega_{3}$    		& 0.2 			& 0.0057 	& 0.0057 & 0.0052 & $0.5360$  & 0.1	\\
$\Omega_{4}$			& 0.1 			& 0.0037 	& 0.0037 & 0.0033 & $0.3704$  & 1.4 	\\
$\Omega_{6}$			& 0.1 			& 0.0056 	& 0.0056 & 0.0051 & $0.4595$  & 6.9 \\
\midrule[.3pt]                                                                                                
Rep (MIP)		& \multicolumn{1}{c}{--}& 0.1395 				& 0.2345                            & 0.1391 				& 0.3636    & 3600	\\
Rep (G-PCCA)		   	& \multicolumn{1}{c}{--}& 0.0163 				& \multicolumn{1}{c}{--}            & 0.0159 				& 0.3754    & $<$ 1	\\
\bottomrule       
\end{tabular} 
\caption{Summary of results for all test instances. The repressilator MIP could not be solved to optimality within the time limit of 3600 seconds. Note that the G-PCCA result in the last line are added as a reference and computed by a different clustering algorithm without guaranteed bound on the objective, see~\cite{WeberFackeldey2015}.}
\label{tab:output}
\end{table}

We analyzed results for the two different experimental settings of Section~\ref{subsec:testset}.
In the \emph{Catalytic Cycle} example, there exist strong metastabilities with a weak non-reversible net flow between them.
Our results show that our method is able to identify non-dominant cycles between metastable clusters.
The results for the second example \emph{Repressilator} show that we are also able to find dominating cyclic structures.
Moreover, our novel MIP formulation provides biologically meaningful clusterings and outperforms state-of-the-art approaches in that sense.

\paragraph{Catalytic Cycle.}
In the first four rows of Table~\ref{tab:output} and in Figure~\ref{fig:PotentialSolutions} one can see the results of the HMC simulations of potential energy surfaces different number of metastabilities ($3$, $4$, and $6$) and a small or larger drift ($0.1$ and $0.2$).
Column ``best obj.'' states the objective function value of the best cycle clustering that was computed and is composed of column ``net flow'' ($\sum_k f(C_k, C_{\phi(k)})$) plus~$0.001$ times column ``coherence'' ($\sum_k g(C_k)$).
Because SCIP could compute proven optimal cycle clusterings for these four instances, the  ``best obj.'' values equal the ``obj.~bound'' values that state the proven upper bound.

\begin{figure}[t]
  \centering
  \ifthenelse{\usecolors = 1}{
    \includegraphics[scale=0.28,trim={2cm 1.3cm 0 0},clip]{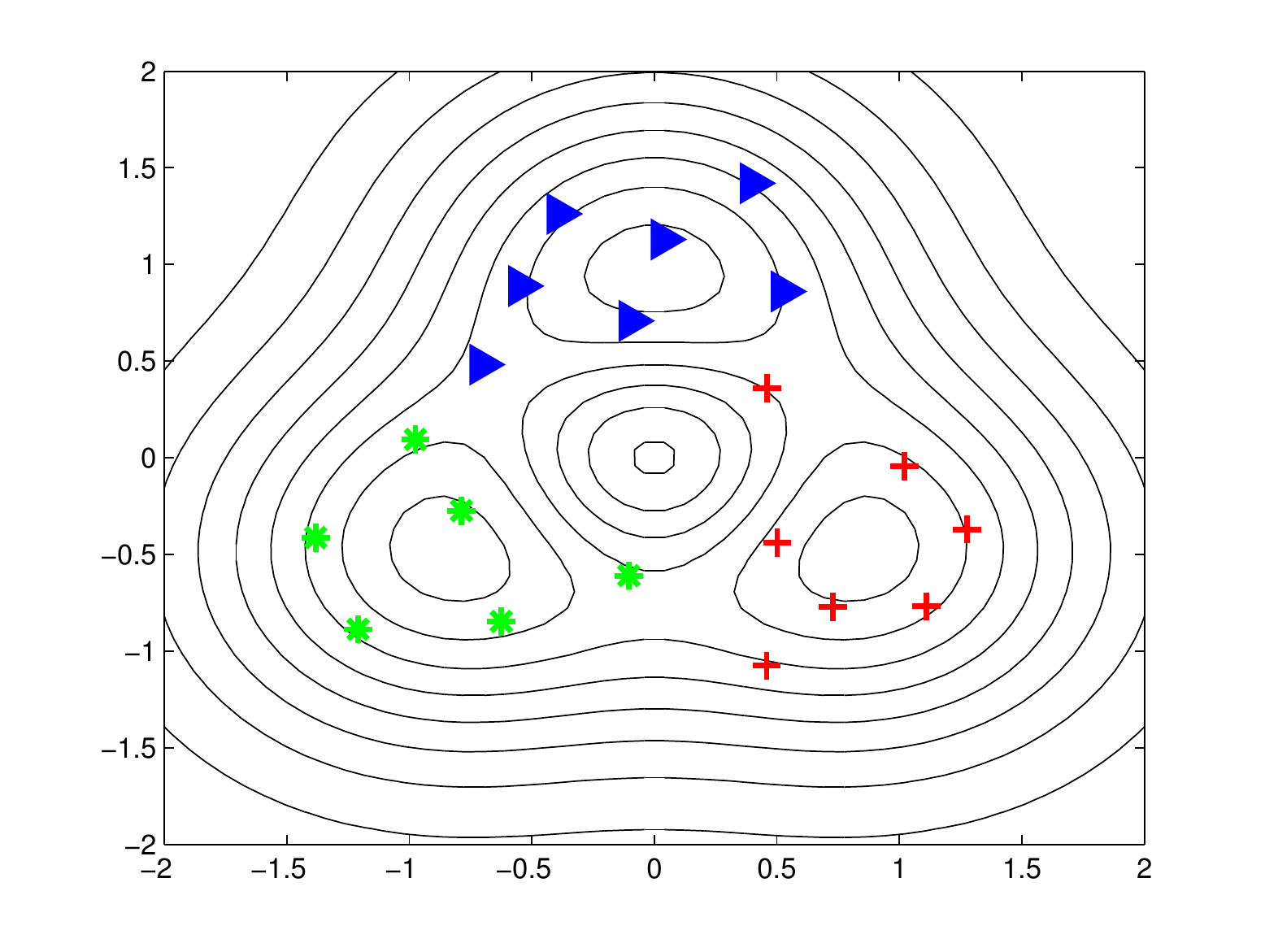}
    \includegraphics[scale=0.28,trim={2cm 1.3cm 0 0},clip]{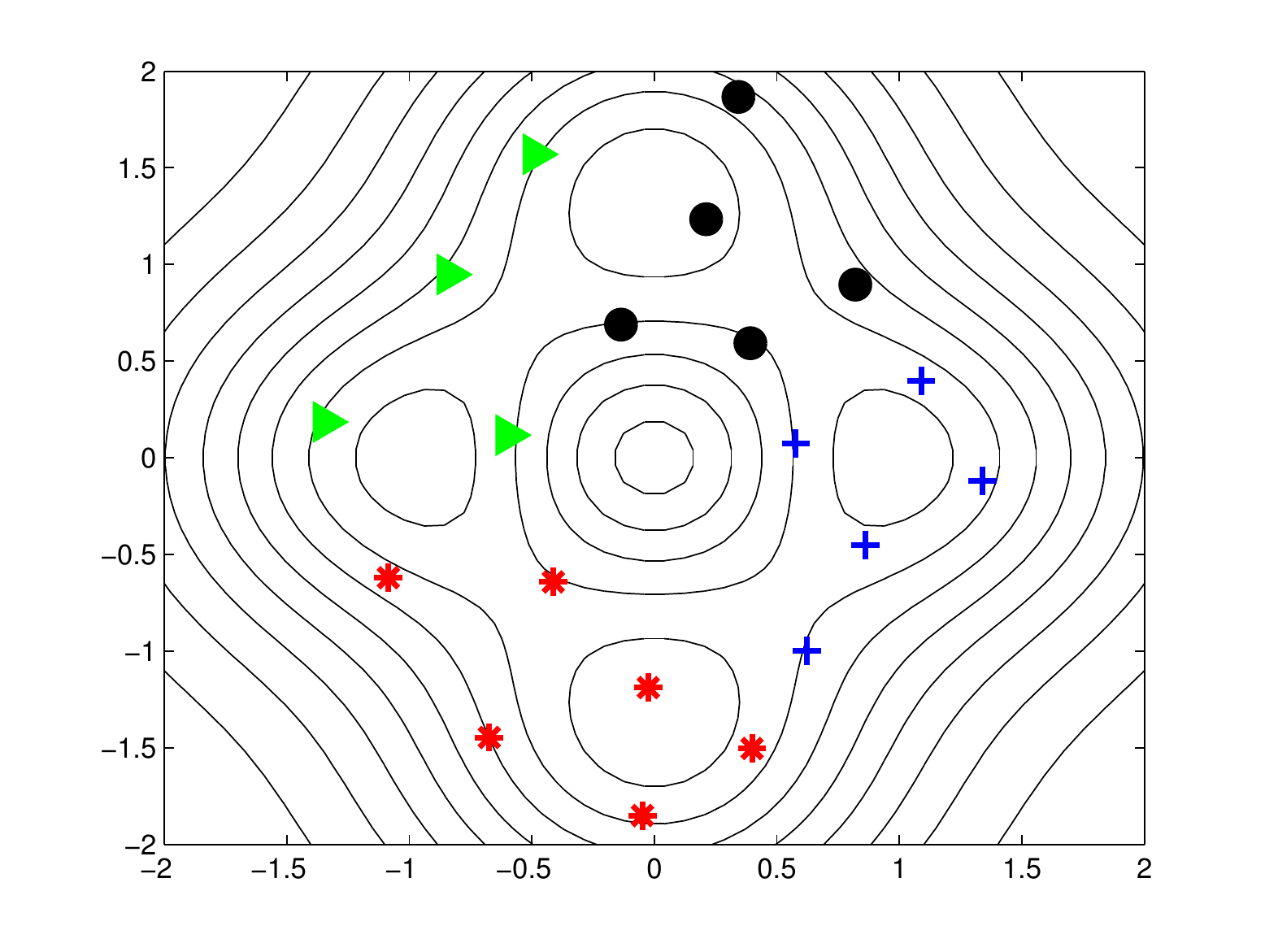}
    \includegraphics[scale=0.28,trim={2cm 1.3cm 0 0},clip]{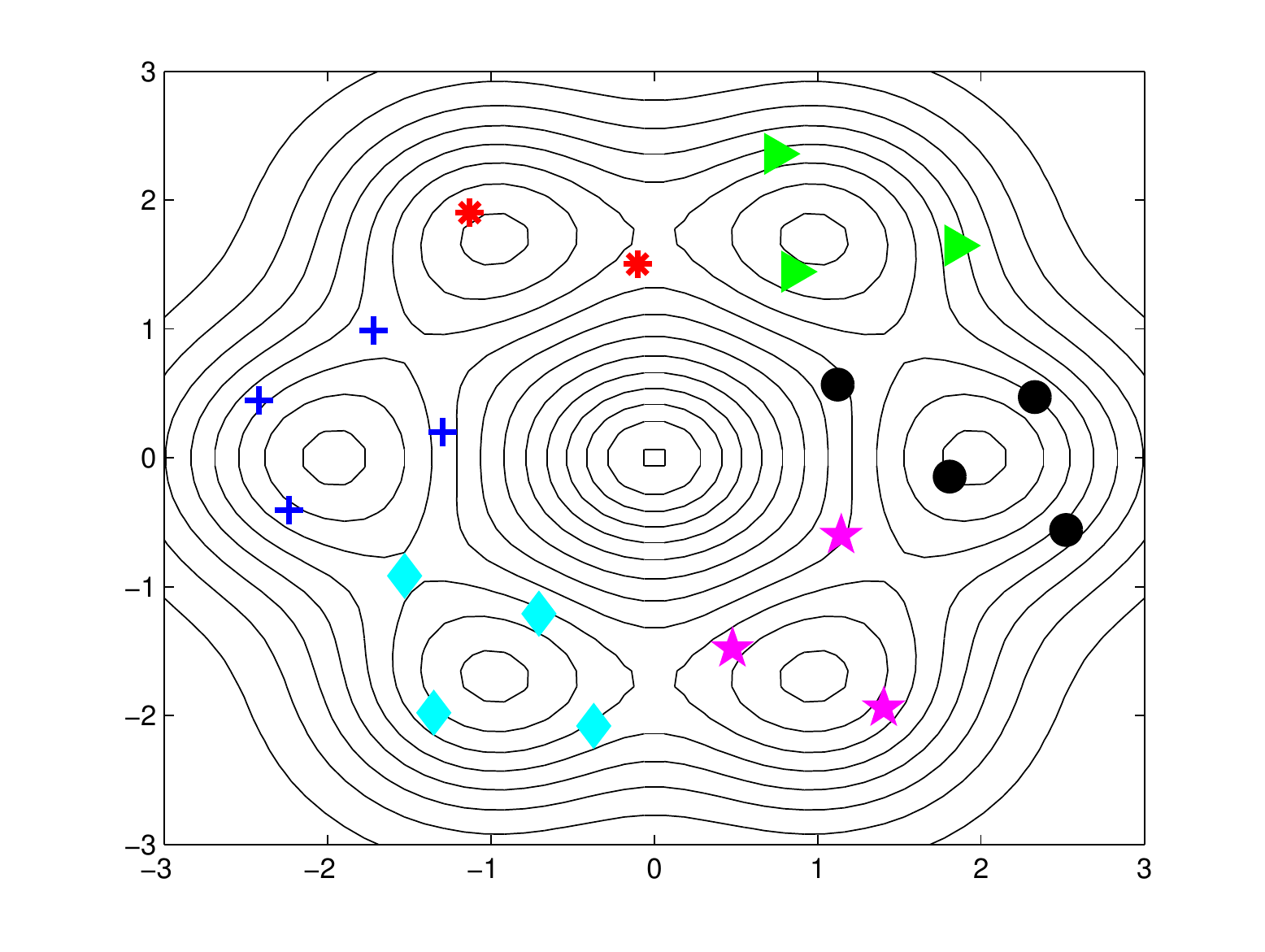}
  }{
      \includegraphics[scale=0.28,trim={2cm 1.3cm 0 0},clip]{Pot3-3-sw.eps}
      \includegraphics[scale=0.28,trim={2cm 1.3cm 0 0},clip]{Pot4-4-sw.eps}
     \includegraphics[scale=0.28,trim={2cm 1.3cm 0 0},clip]{Pot6-6-sw.eps}
  }
  \caption{Visualization of the solutions obtained with a 3-cycle, 4-cycle, and 6-cycle clustering for the synthetic potentials $\Omega_{3}$, $\Omega_{4}$, and $\Omega_{6}$, respectively.}
\label{fig:PotentialSolutions}
\end{figure}

In all four cases our approach was able to correctly identify the metastabilities and the direction of the drift.

Furthermore, we can observe that the coherence within the clusters is large although $\alpha$ was chosen small.
Note that it equals~$1$ minus the probability of seeing transitions between clusters.
By construction, the net flow is comparatively small, but could still be detected consistently.
The next example shows that our method is also able to analyze systems with a fast and productive cycle.

\paragraph{Repressilator.}
For the repressilator model, the last two rows of Table~\ref{tab:output} compare the cycle clustering solution (MIP) with a solution of the spectral clustering algorithm (G-PCCA) explained in Section~\ref{sec:background} and \cite{DjurdjevacConradWeberSchuette2016,WeberFackeldey2015}.
Note that here the ``drift'' column does not apply and G-PCCA was not designed with this objective function in mind and by its nature does not compute any proven objective bound.
In this sense we want to emphasize that this comparison does not provide any kind of benchmark.

\begin{figure}[t]
\begin{subfigure}{\textwidth}
	\centering
	\ifthenelse{\usecolors = 1}{
	  \scalebox{0.8}{
	  
	  \begin{tikzpicture}  
  \begin{scope}
      \begin{axis}[
          symbolic x coords={$p_{A}$,$m_{A}$, $p_{B}$,$m_{B}$, $p_{C}$,  $m_{C}$},
          xtick=data,
          ymin=0,
          ymax=15,
          width=5cm,height=5cm
        ]
        \addplot[ybar,fill=brown] coordinates {
          ($p_{A}$, 10.1941)
          ($m_{A}$, 0)
          ($p_{B}$, 0)
          ($m_{B}$, 0)
          ($p_{C}$, 0)
          ($m_{C}$, 0)
        };
        \addplot[ybar,fill=brown!40] coordinates {
          ($m_{A}$, 11.5285)
        };
        \addplot[ybar,fill=darkgreen] coordinates {
          ($p_{B}$, 10.0204)
        };
        \addplot[ybar,fill=darkgreen!40] coordinates {
          ($m_{B}$, 10.7842)
        };
        \addplot[ybar,fill=cyan] coordinates {
          ($p_{C}$, 9.6986)
        };
        \addplot[ybar,fill=cyan!40] coordinates {
          ($m_{C}$, 7.8668)
        };
      \end{axis}
  \end{scope}
  
   \begin{scope}[xshift=4cm]
    \begin{axis}[
            symbolic x coords={$p_{A}$,$m_{A}$, $p_{B}$,$m_{B}$, $p_{C}$,  $m_{C}$},
            xtick=data,
            ymin=0,
            ymax=15,
            width=5cm,height=5cm,
            yticklabels={,,}
          ]
          \addplot[ybar,fill=brown] coordinates {
            ($p_{A}$, 9.6431)
            ($m_{A}$, 0)
            ($p_{B}$, 0)
            ($m_{B}$, 0)
            ($p_{C}$, 0)
            ($m_{C}$, 0)
          };
          \addplot[ybar,fill=brown!40] coordinates {
            ($m_{A}$, 7.7293)
          };
          \addplot[ybar,fill=darkgreen] coordinates {
            ($p_{B}$, 10.4012)
          };
          \addplot[ybar,fill=darkgreen!40] coordinates {
            ($m_{B}$, 11.3031)
          };
          \addplot[ybar,fill=cyan] coordinates {
            ($p_{C}$, 9.9418)
          };
          \addplot[ybar,fill=cyan!40] coordinates {
            ($m_{C}$, 10.9544)
          };
        \end{axis}
  \end{scope}
  
  \begin{scope}[xshift=8cm]
   \begin{axis}[
           symbolic x coords={$p_{A}$,$m_{A}$, $p_{B}$,$m_{B}$, $p_{C}$,  $m_{C}$},
           xtick=data,
           ymin=0,
           ymax=15,
           width=5cm,height=5cm,
           yticklabels={,,}
         ]
        \addplot[ybar,fill=brown] coordinates {
          ($p_{A}$, 10.1248)
          ($m_{A}$, 0)
          ($p_{B}$, 0)
          ($m_{B}$, 0)
          ($p_{C}$, 0)
          ($m_{C}$, 0)
        };
        \addplot[ybar,fill=brown!40] coordinates {
          ($m_{A}$,   11.4076)
        };
        \addplot[ybar,fill=darkgreen] coordinates {
          ($p_{B}$, 9.6136)
        };
        \addplot[ybar,fill=darkgreen!40] coordinates {
          ($m_{B}$,  7.6979)
        };
        \addplot[ybar,fill=cyan] coordinates {
          ($p_{C}$, 10.3663)
        };
        \addplot[ybar,fill=cyan!40] coordinates {
          ($m_{C}$,   10.7210)
        };
   \end{axis}
  \end{scope}

   \begin{scope}[xshift=4cm, yshift=3.5cm]
   \draw[-,,line width=1pt] (-2,0) -- (-2,0.5);
   \draw[-,line width=1pt] (-2,0.5) --node[pos=.5,fill=white,circle] {$\NetFlow_1$} (1, 0.5);
   \draw[->,line width=1pt] (1,0.5) -- (1, 0);
   \end{scope}
   
   \begin{scope}[xshift=8cm, yshift=3.5cm]
   \draw[-,,line width=1pt] (-2,0) -- (-2,0.5);
   \draw[-,line width=1pt] (-2,0.5) --node[pos=.5,fill=white,circle] {$\NetFlow_2$} (1, 0.5);
   \draw[->,line width=1pt] (1,0.5) -- (1, 0);
   \end{scope}
   
   \begin{scope}[xshift=4cm]
   \draw[->,line width=1pt] (-2,-1) -- (-2, -0.5);
   \draw[-,line width=1pt] (-2,-1) --node[pos=.5,fill=white,circle] {$\NetFlow_3$} (5, -1);
   \draw[-,line width=1pt] (5,-1) -- (5, -0.5);
   \end{scope}
  
  \end{tikzpicture}

	  }
	}{
	  \scalebox{0.8}{\input{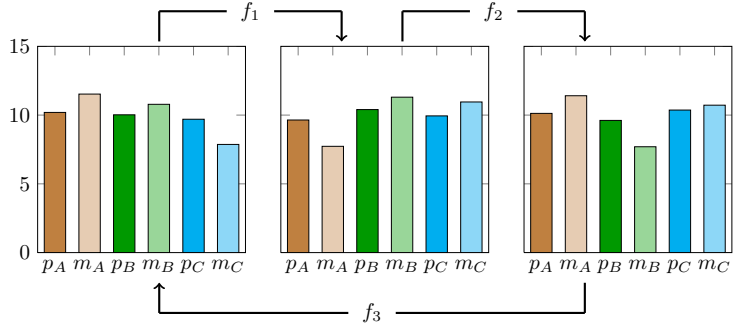}}
	}
	\caption{Solution obtained by G-PCCA.}
	\label{fig:BarRepressilatorGP}
\end{subfigure}
\bigskip

\begin{subfigure}{\textwidth}
	\centering
	\ifthenelse{\usecolors = 1}{
	  \scalebox{0.8}{
	  
	  \begin{tikzpicture}
  
  \begin{scope}[xshift=4cm]
      \begin{axis}[
          symbolic x coords={$p_{A}$,$m_{A}$, $p_{B}$,$m_{B}$, $p_{C}$,  $m_{C}$},
          xtick=data,
          ymin=0,
          ymax=15,
          width=5cm,height=5cm
        ]
        \addplot[ybar,fill=brown] coordinates {
          ($p_{A}$, 10.1412)
          ($m_{A}$, 0)
          ($p_{B}$, 0)
          ($m_{B}$, 0)
          ($p_{C}$, 8.2381)
          ($m_{C}$, 0)
         };
        \addplot[ybar,fill=brown!40] coordinates {
          ($m_{A}$, 8.2479)
        };
        \addplot[ybar,fill=darkgreen] coordinates {
          ($p_{B}$, 13.0237)
         };
        \addplot[ybar,fill=darkgreen!40] coordinates {
          ($m_{B}$, 13.8828)
        };
        \addplot[ybar,fill=cyan] coordinates {
          ($p_{C}$, 8.2381)
         };
        \addplot[ybar,fill=cyan!40] coordinates {
          ($m_{C}$, 9.1041)
        };
      \end{axis}
  \end{scope}
  
   \begin{scope}[xshift=8cm]
    \begin{axis}[
            symbolic x coords={$p_{A}$,$m_{A}$, $p_{B}$,$m_{B}$, $p_{C}$,  $m_{C}$},
            xtick=data,
            ymin=0,
            ymax=15,
            width=5cm,height=5cm,
            yticklabels={,,}
          ]
          \addplot[ybar,fill=brown] coordinates {
            ($p_{A}$, 7.2559)
            ($m_{A}$, 0)
            ($p_{B}$, 0)
            ($m_{B}$, 0)
            ($p_{C}$, 0)
            ($m_{C}$, 0)
          };
          \addplot[ybar,fill=brown!40] coordinates {
            ($m_{A}$, 9.0430)
          };
          \addplot[ybar,fill=darkgreen] coordinates {
            ($p_{B}$, 9.1273)
          };
          \addplot[ybar,fill=darkgreen!40] coordinates {
            ($m_{B}$, 8.6249)
          };
          \addplot[ybar,fill=cyan] coordinates {
            ($p_{C}$, 12.4660)
          };
          \addplot[ybar,fill=cyan!40] coordinates {
            ($m_{C}$, 13.3949)
          };
        \end{axis}
  \end{scope}
  
  \begin{scope}[xshift=0cm]
   \begin{axis}[
           symbolic x coords={$p_{A}$,$m_{A}$, $p_{B}$,$m_{B}$, $p_{C}$,  $m_{C}$},
           xtick=data,
           ymin=0,
           ymax=15,
           width=5cm,height=5cm,
           yticklabels={,,}
         ]
          \addplot[ybar,fill=brown] coordinates {
            ($p_{A}$, 12.3816)
            ($m_{A}$, 0)
            ($p_{B}$, 0)
            ($m_{B}$, 0)
            ($p_{C}$, 0)
            ($m_{C}$, 0)
          };
          \addplot[ybar,fill=brown!40] coordinates {
            ($m_{A}$, 12.6960)
          };
          \addplot[ybar,fill=darkgreen] coordinates {
            ($p_{B}$, 7.5718)
          };
          \addplot[ybar,fill=darkgreen!40] coordinates {
            ($m_{B}$, 7.2500)
          };
          \addplot[ybar,fill=cyan] coordinates {
            ($p_{C}$, 9.6104)
          };
          \addplot[ybar,fill=cyan!40] coordinates {
            ($m_{C}$, 7.6717)
          };
   \end{axis}
  \end{scope}

   \begin{scope}[xshift=4cm, yshift=3.5cm]
   \draw[-,,line width=1pt] (-2,0) -- (-2,0.5);
   \draw[-,line width=1pt] (-2,0.5) --node[pos=.5,fill=white,circle] {$\NetFlow_1$} (1, 0.5);
   \draw[->,line width=1pt] (1,0.5) -- (1, 0);
   \end{scope}
   
   \begin{scope}[xshift=8cm, yshift=3.5cm]
   \draw[-,,line width=1pt] (-2,0) -- (-2,0.5);
   \draw[-,line width=1pt] (-2,0.5) --node[pos=.5,fill=white,circle] {$\NetFlow_2$} (1, 0.5);
   \draw[->,line width=1pt] (1,0.5) -- (1, 0);
   \end{scope}
   
   \begin{scope}[xshift=4cm]
   \draw[->,line width=1pt] (-2,-1) -- (-2, -0.5);
   \draw[-,line width=1pt] (-2,-1) --node[pos=.5,fill=white,circle] {$\NetFlow_3$} (5, -1);
   \draw[-,line width=1pt] (5,-1) -- (5, -0.5);
   \end{scope}
   
  \end{tikzpicture}

	  }
	}{
	  \scalebox{0.8}{\input{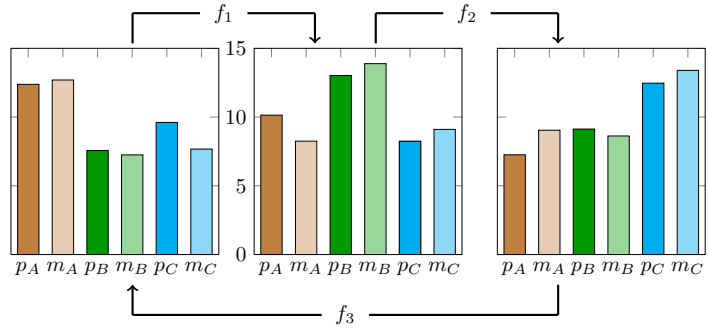}}
	}
	\caption{Solution obtained by cycle clustering.}
	\label{fig:BarRepressilator}
\end{subfigure}
\caption{Average values of the concentrations in each cluster.}
\end{figure}

The resulting MIP problem was significantly larger and computationally more challenging than for the synthetic instances, and SCIP could not solve it to optimality within one hour.
Nevertheless, the best clustering solution returned by SCIP exhibits both large net flow and coherence.
Furthermore, the objective bound guarantees that it is at most a factor of~$1.8$ from the objective of an optimal cycle clustering.

To analyze the clusterings in more detail, consider the projected transition matrices in the notation of Chapter~\ref{subsec:clusteringmodel},
\begin{align*}
\overline{W}_{\text{MIP}}=\begin{pmatrix}
    0.126 &   0.134  &  0.075\\
    0.088 &   0.150  &  0.123\\
    0.121 &   0.077  &  0.107
\end{pmatrix}
\end{align*}
with
\begin{align*}
\overline{W}_{\text{MIP}}-\overline{W}_{\text{MIP}}^{T}= 0.046\cdot
\begin{pmatrix}
0 & 1 & -1\\ -1 & 0 & 1 \\ 1 & -1  &  0
\end{pmatrix}.
\end{align*}
Here, the cycle $\Cluster_1\rightarrow \Cluster_2 \rightarrow \Cluster_3$ can be clearly identified.  In comparison to this, the matrix
\begin{align*}
\overline{W}_{\text{G-PCCA}}=\begin{pmatrix}
    0.082 &   0.090 &   0.103\\
    0.084 &   0.113 &   0.123\\
    0.108 &   0.117 &   0.180
\end{pmatrix}
\end{align*}
does not have such a clear cyclic order and the net flow is one order of magnitude smaller,
\begin{align*}
\overline{W}_{\text{G-PCCA}}-\overline{W}_{\text{G-PCCA}}^{T}=
0.005\cdot\begin{pmatrix}
0 & 1 & -1\\ -1 & 0 & 1 \\ 1 & -1  &  0
\end{pmatrix}.
\end{align*}

In Figure~\ref{fig:BarRepressilatorGP}, we have plotted the arithmetic averages
of the protein and mRNA concentrations in each cluster both for the G-PCCA clustering and the cycle clustering.
In the cycle clustering solution, one can clearly identify a peak of $p_A,m_A$ in cluster~$\Cluster_1$, a peak of $p_B,m_B$ in cluster~$\Cluster_2$, and a peak of $p_C,m_C$ in cluster~$\Cluster_3$.
In contrast, the protein and mRNA concentrations in the G-PCCA clusters seem more uniform and---besides a slightly decreased concentration of $m_C$,  $m_A$, and $m_B$ in cluster $\Cluster_1$, $\Cluster_2$, and $\Cluster_3$, respectively---make it difficult to detect a particular structure.

In this sense, cycle clustering succeeds in separating states in a biologically meaningful way, using only simulation data, without additional knowledge about the defining dynamical system.
Each identified cluster of the system corresponds to a 
biological entity of a protein together with its mRNA.
In constrast, the 
G-PCCA solution does not seem to be interpretable in terms of (separate) 
biological entities. The high productivity of the system is not 
reflected by the objective value of the G-PCCA solution, while the 
MIP solution accounts well for the cyclic nature of the repressilator.

\paragraph{MIP Performance.}
As can be seen from the running time in the last column of~Table~\ref{tab:output}, the improved clustering solutions come at the price of increased computing times.
While G-PCCA took less than a second, SCIP could not solve the MIP problem to optimality within the time limit of one hour.
This is not surprising given Theorem~\ref{thm:complexity}, which proves that cycle clustering is a hard combinatorial optimization problem.

However, good cycle clustering solutions are usually computed very early during the solution process and for the repressilator solution the dual bound proves that the best clustering w.r.t.\ our objective function can be at most $2.6$~times as good as the solution stated in Table~\ref{tab:output}.
In future research, we will focus on improving the performance of MIP solvers for cycle clustering by further dedicated techniques.

\section{Conclusion}
\label{sec:conclusion}

Many non-reversible biological processes seem to be mainly reversible on a small time-scale. 
In order to identify the global cyclic behavior of the system, we 
formulated an optimization problem for partitioning the state 
space into certain cluster which are ``visited'' in a 
non-reversible manner.  Standard approaches which use spectral 
information of the transition matrix $P$ are constructed such that they 
find dominant cycles or the strongest metastabilities, but they do not 
account for non-reversible cycles if these cycles are hidden, i.e., not 
dominant.

We prove that the new clustering method amounts to solving an $\NP$-hard combinatorial optimization problem.
However, computational experiments show that our solution strategy, which uses a mixed-integer programming formulation,
effectively finds optimal or near-optimal clusterings. 
The results for a genetic regulatory network demonstrate that the 
identified clusters are meaningful in the biological context. This clustering was not 
found by standard or spectral approaches.
One reason for that could be, that the problem of finding a global cycle
with maximal net flow is a more complex problem than determining a
dominant cycle with the help of spectral analysis, i.e. Schur-decomposition.

\subsection*{Acknowledgements}

The work for this article has been partly conducted within the Research Campus Modal
funded by the German Federal Ministry of Education and Research (fund number 05M14ZAM).

\bibliographystyle{abbrv}
\bibliography{Bibliography}

\begin{thebibliography}{10}

\bibitem{achterberg2008constraint}
T.~Achterberg, T.~Berthold, T.~Koch, and K.~Wolter.
\newblock Constraint integer programming: A new approach to integrate cp and
  mip.
\newblock In {\em International Conference on Integration of Artificial
  Intelligence (AI) and Operations Research (OR) Techniques in Constraint
  Programming}, pages 6--20. Springer, 2008.

\bibitem{applegate1998solution}
D.~Applegate, R.~Bixby, W.~Cook, and V.~Chv{\'a}tal.
\newblock {\em On the solution of travelling salesman problems}.
\newblock Universit{\"a}t Bonn. Institut f{\"u}r {\"O}konometrie und Operations
  Research, 1998.

\bibitem{BrownPandeNoe2014}
G.~Bowman, V.~Pande, and F.~No\'e, editors.
\newblock {\em An Introduction to Markov State Models and Their Application to
  Long Timescale Molecular Simulation}, volume 797 of {\em Advances in
  Experimental Medicine and Biology}. Springer Berlin Heidelberg, 2014.

\bibitem{Brooks2011}
S.~Brooks, A.~Gelman, G.~Jones, and X.-L. Meng.
\newblock {\em Handbook of Markov Chain Monte Carlo}.
\newblock CRC press, 2011.

\bibitem{dahlhaus1992complexity}
E.~Dahlhaus, D.~S. Johnson, C.~H. Papadimitriou, P.~D. Seymour, and
  M.~Yannakakis.
\newblock The complexity of multiway cuts.
\newblock In {\em Proceedings of the twenty-fourth annual ACM symposium on
  Theory of computing}, pages 241--251. ACM, 1992.

\bibitem{dakin1965tree}
R.~J. Dakin.
\newblock A tree-search algorithm for mixed integer programming problems.
\newblock {\em The computer journal}, 8(3):250--255, 1965.

\bibitem{DHFS_00}
P.~Deuflhard, W.~Huisinga, A.~Fischer, and C.~Sch{\"u}tte.
\newblock Identification of almost invariant aggregates in reversible nearly
  uncoupled markov chains.
\newblock {\em Linear Algebra and its Applications}, 315(1-3):39 -- 59, 2000.

\bibitem{DW_05}
P.~Deuflhard and M.~Weber.
\newblock Robust perron cluster analysis in conformation dynamics.
\newblock {\em Linear Algebra and its Applications}, 398:161 -- 184, 2005.
\newblock Special Issue on Matrices and Mathematical Biology.

\bibitem{DjurdjevacConradBanischSchuette2014}
N.~Djurdjevac-Conrad, R.~Banisch, and C.~Sch{\"u}tte.
\newblock Modularity of directed networks: Cycle decomposition approach.
\newblock {\em Journal of Computational Dynamics 2 (2015) pp. 1-24}, 2014.

\bibitem{DjurdjevacConradWeberSchuette2016}
N.~Djurdjevac-Conrad, M.~Weber, and C.~Sch{\"u}tte.
\newblock Finding dominant structures of nonreversible markov processes.
\newblock {\em Multiscale Modeling and Simulation}, 2016.
\newblock accepted for publication.

\bibitem{ElowitzLeibler2000}
M.~Elowitz and S.~Leibler.
\newblock A synthetic oscillatory network of transcriptional regulators.
\newblock {\em Nature}, 403(20):335--338, 2000.

\bibitem{FackeldeyWeber2014}
K.~Fackeldey and M.~Weber.
\newblock Local refinements in classical molecular dynamics simulations.
\newblock {\em 2nd International Conference on Mathematical Modeling in
  Physical Sciences 2013, Journal of Physics: Conference Series 490:012016},
  2014.

\bibitem{fortet1960algebre}
R.~Fortet.
\newblock L'algebre de boole et ses applications en recherche
  op{\'e}rationnelle.
\newblock {\em Trabajos de Estadistica y de Investigaci{\'o}n Operativa},
  11(2):111--118, 1960.

\bibitem{GamrathFischerGallyetal.2016}
G.~Gamrath, T.~Fischer, T.~Gally, A.~M. Gleixner, G.~Hendel, T.~Koch, S.~J.
  Maher, M.~Miltenberger, B.~M{\"u}ller, M.~E. Pfetsch, C.~Puchert,
  D.~Rehfeldt, S.~Schenker, R.~Schwarz, F.~Serrano, Y.~Shinano, S.~Vigerske,
  D.~Weninger, M.~Winkler, J.~T. Witt, and J.~Witzig.
\newblock {The SCIP Optimization Suite 3.2}.
\newblock ZIB-Report 15-60, Zuse Institute Berlin, 2016.

\bibitem{GareyJohnson1979}
M.~R. Garey and D.~S. Johnson.
\newblock {\em Computers and Intractability: A Guide to the Theory of
  NP-Completeness}.
\newblock W. H. Freeman \& Co., New York, NY, USA, 1979.

\bibitem{hagen1992new}
L.~Hagen and A.~B. Kahng.
\newblock New spectral methods for ratio cut partitioning and clustering.
\newblock {\em IEEE transactions on computer-aided design of integrated
  circuits and systems}, 11(9):1074--1085, 1992.

\bibitem{KernighanLin1970}
B.~Kernighan and S.~Lin.
\newblock {An Efficient Heuristic Procedure for Partitioning Graphs}.
\newblock {\em The Bell Systems Technical Journal}, 49(2), 1970.

\bibitem{land1960automatic}
A.~H. Land and A.~G. Doig.
\newblock An automatic method of solving discrete programming problems.
\newblock {\em Econometrica: Journal of the Econometric Society}, pages
  497--520, 1960.

\bibitem{LN_84}
R.~Lidl and H.~Niederreiter.
\newblock {\em Finite Fields}.
\newblock Cambridge University Press, Cambridge, 1984.

\bibitem{N_88}
H.~Niederreiter.
\newblock Low-discrepancy and low-dispersion sequences.
\newblock {\em Journal of Number Theory}, (30):51--78, 1988.

\bibitem{NielsenWeber2015}
A.~Nielsen and M.~Weber.
\newblock Computing the nearest reversible markov chain.
\newblock {\em Numerical Linear Algebra with Applications}, 22(3):483 -- 499,
  2015.

\bibitem{PSK_11}
J.-H. Prinz, H.~Wu, M.~Sarich, B.~Keller, M.~Senne, M.~Held, J.~D. Chodera,
  C.~Sch{\"u}tte, and F.~No{\'e}.
\newblock Markov models of molecular kinetics: Generation and validation.
\newblock {\em The Journal of Chemical Physics}, 134(17), 2011.

\bibitem{savelsbergh1994preprocessing}
M.~W. Savelsbergh.
\newblock Preprocessing and probing techniques for mixed integer programming
  problems.
\newblock {\em ORSA Journal on Computing}, 6(4):445--454, 1994.

\bibitem{schrijver2003combinatorial}
A.~Schrijver.
\newblock {\em Combinatorial Optimization: Polyhedra and Efficiency},
  volume~24.
\newblock Springer, 2003.

\bibitem{shi2000normalized}
J.~Shi and J.~Malik.
\newblock Normalized cuts and image segmentation.
\newblock {\em IEEE Transactions on pattern analysis and machine intelligence},
  22(8):888--905, 2000.

\bibitem{vonLuxburg2006}
U.~von Luxburg.
\newblock Tutorial on spectral clustering.
\newblock {\em Max Planck Institute for Biological Cybernetics, Technical
  Report No. 149, T{\"u}ubingen}, 2006.

\bibitem{Weber2011}
M.~Weber.
\newblock {\em A Subspace Approach to Molecular Markov State Models via a New
  Infinitesimal Generator}.
\newblock Habilitation thesis, FU Berlin, 2011.

\bibitem{WeberFackeldey2015}
M.~Weber and K.~Fackeldey.
\newblock {G-PCCA}: Spectral clustering for non-reversible markov chains.
\newblock ZIB-Report 15-35, Zuse Institute Berlin, 2015.

\end{thebibliography}

\end{document}